\documentclass[conference]{IEEEtran}

\usepackage{graphicx}
\usepackage{booktabs}
\usepackage{multirow}
\usepackage{tabularx}
\usepackage[noend]{algpseudocode}
\usepackage[ruled]{algorithm}
\usepackage{amsmath}
\usepackage{amssymb}
\usepackage{xcolor}
\usepackage[hidelinks=true]{hyperref}
\usepackage{cite}
\usepackage{microtype}
\usepackage{cite}
\usepackage{paralist}
\usepackage{subfigure}
\usepackage[flushleft]{threeparttable}

\usepackage[keeplastbox]{flushend}

\newcommand{\ignore}[1]{}

\newcommand{\proc}[1]{\textsc{#1}}
\newcommand{\textrel}[1]{\textsc{#1}}
\newcommand{\IntAbs}[1]{\textsf{IntAbs}}
\newcommand{\Name}[1]{\textsf{IntAbs}}
\newtheorem{theorem}{Theorem}

\algrenewcommand\algorithmicindent{1.0em}%

\newcommand\tfunc[0]{\textsc{tfunc}}

\definecolor{mygreen}{rgb}{0,0.6,0}
\definecolor{mygray}{rgb}{0.5,0.5,0.5}
\definecolor{mymauve}{rgb}{0.58,0,0.82}
\definecolor{dkgreen}{rgb}{0,0.6,0}
\definecolor{darkblue}{rgb}{0.0, 0.0, 0.55}

\makeatletter
\let\OldStatex\Statex
\renewcommand{\Statex}[1][3]{%
  \setlength\@tempdima{\algorithmicindent}%
  \OldStatex\hskip\dimexpr#1\@tempdima\relax}
\makeatother

\makeatletter
\newcommand\fs@betterruled{%
  \def\@fs@cfont{\bfseries}\let\@fs@capt\floatc@ruled
  \def\@fs@pre{\vspace*{5pt}\hrule height.8pt depth0pt \kern2pt}%
  \def\@fs@post{\kern2pt\hrule\relax}%
  \def\@fs@mid{\kern2pt\hrule\kern2pt}%
  \let\@fs@iftopcapt\iftrue}
\floatstyle{betterruled}
\restylefloat{algorithm}
\makeatother

\begin{document}
%
\title{Modular Verification of Interrupt-Driven Software}

\author{
\IEEEauthorblockN{Chungha Sung}
\IEEEauthorblockA{
University of Southern California\\
Los Angeles, CA, USA
}
\and
\IEEEauthorblockN{Markus Kusano}
\IEEEauthorblockA{
Virginia Tech\\
Blacksburg, VA, USA
}
\and
\IEEEauthorblockN{Chao Wang}
\IEEEauthorblockA{
University of Southern California\\
Los Angeles, CA, USA
}
}

\maketitle

\begin{abstract}

Interrupts have been widely used in safety-critical computer systems
to handle outside stimuli and interact with the hardware, but
reasoning about interrupt-driven software remains a difficult task.
Although a number of static verification techniques have been proposed
for interrupt-driven software, they often rely on constructing a
monolithic verification model.  Furthermore, they do not precisely
capture the complete execution semantics of interrupts such as nested
invocations of interrupt handlers.
To overcome these limitations, we propose an \emph{abstract
interpretation} framework for static verification of interrupt-driven
software that
first analyzes each interrupt handler in isolation as if it
were a sequential program, and
then propagates the result to other interrupt handlers.  This
\emph{iterative} process continues until results from all interrupt handlers
reach a fixed point. Since our method never constructs the global model, 
it avoids the up-front blowup in model construction that hampers existing, non-modular,
verification techniques.
We have evaluated our method on 35 interrupt-driven applications with
a total of 22,541 lines of code.  Our results show the method
is able to quickly and more accurately analyze the behavior of interrupts.

\end{abstract}

%
\IEEEpeerreviewmaketitle

\section{Introduction}

Interrupts have been widely used in safety-critical embedded computing
systems, information processing systems, and mobile systems to interact
with hardware and respond to outside stimuli in a timely manner.
However, since interrupts may arrive non-deterministically at any
moment to preempt the normal computation, they are difficult for
developers to reason about.  The situation is further exacerbated by
the fact that interrupts often have different priority levels:
high-priority interrupts may preempt low-priority interrupts but not
vice versa, and interrupt handlers may be executed in a nested
fashion.  Overall, methods and tools for accurately modeling the
semantics of interrupt-driven software are still lacking.

Broadly speaking, existing techniques for analyzing interrupts fall
into two categories.
The first category consists of techniques based on
testing~\cite{Regehr05, Higashi10}, which rely on executing the
program under various interrupt invocation sequences.  Since it is
often practically infeasible to cover all combinations of interrupt
invocations, testing will miss important bugs. 
The second category consists of static verification techniques such as
model checking~\cite{Bucur11, Schlich09, WuWCDW13, Vortler15,
Kroening15}, which rely on constructing and analyzing a formal model.
During the modeling process, interrupt-related behaviors such as
preemption are considered.  Unfortunately, existing tools such as
iCBMC~\cite{Kroening15} need to bound the execution depth to remain
efficient, which means shallow bugs can be detected quickly, but
these tools cannot prove the absence of bugs.

In this paper, we propose a static verification tool geared
toward proving the absence of bugs based on \emph{abstract
interpretation}~\cite{Cousot77}.  The main advantage of abstract
interpretation is the sound approximation of complex constructs such
as loops, recursions and numerical computations. However, although
abstract interpretation techniques have been successfully applied to
sequential~\cite{CousotCFMMMR05} and multithreaded
software~\cite{Mine14}, they have not been able to precisely model the
semantics of interrupt-driven software.

At the high level, interrupts share many similarities with threads,
e.g., both interrupt handlers and thread routines may be regarded as
sequential programs communicating with others via the shared memory.  However,
there are major differences in the way they interleave.  For example,
in most of the existing verification tools, threads are allowed to
freely preempt each other's execution.  In contrast, interrupts often
have various levels of priority: high-priority interrupts can preempt
low-priority interrupts but not vice versa.  Furthermore, interrupts
with the same level of priority cannot preempt each other. Thus,
 the behavior manifested by interrupts has to be viewed as a
subset of the behavior manifested by threads.

To accurately analyze the behavior of interrupts, we
develop \Name{}, an \emph{iterative} abstract interpretation
framework for interrupt-driven software.
That is, the framework always analyzes each
interrupt handler in isolation before propagating the result to other
interrupt handlers and the \emph{per interrupt} analysis is iterated
until results on all interrupt handlers stabilize, i.e., they reach
a \emph{fixed point}.  Thus, in contrast to traditional techniques, it
never constructs the monolithic verification model that often causes
exponential blowup up front. Due to this reason, our method is
practically more efficient than these traditional verification techniques.

The \Name{} framework also differs from prior techniques for statically
analyzing interrupt-driven software, such as the source-to-source
transformation-based testing approach proposed by
Regehr~\cite{Regehr07}, the sequentialization 
approach used by Wu et al.~\cite{WuCMDW16}, and the model
checking technique implemented in iCBMC~\cite{Kroening15}.  For
example, none of these existing techniques can soundly handle infinite
loops, nested invocations of interrupts, or prove the absence of bugs. 
Although some prior abstract interpretation techniques~\cite{Mine17} 
over-approximate of the interrupt behavior, they are either
non-modular or too inaccurate, e.g., by allowing too many
infeasible \emph{store-to-load} data flows between interrupts.  
In contrast, our approach precisely models the preemptive scheduling
of interrupts to identify  apparently-infeasible data
flows. As shown in Fig.~\ref{fig:overview}, by pruning away these infeasible
data flows, we can drastically improve the accuracy of the overall
analysis.

\Name{} provides not only a more accurate modeling of the
interrupt semantics but also a more efficient abstract
interpretation framework.  We have implemented \Name{} in a static
analysis tool for C/C++ programs, which uses Clang/LLVM~\cite{Adve03}
as the front-end, Apron \cite{Jeannet09} for implementing the
numerical abstract domains, and $\mu Z$~\cite{Hoder11} for checking
the feasibility of data flows between interrupts.  We evaluated \Name{}
on 35 interrupt-driven applications with a total of 22,541
lines of C code.  Our experimental results show that
\Name{} can efficiently as well as more accurately analyze the 
behavior of interrupts by removing a large number of infeasible
data flows between interrupts.

In summary, the main contributions of our work are:

\begin{itemize}

\item 
A new abstract interpretation framework for
conducting static verification of interrupt-driven programs.

\item 
A method for soundly and efficiently identifying and pruning
infeasible data flows between interrupts.

\item 
The implementation and experimental evaluation on a large number of
benchmark programs to demonstrate the effectiveness of the proposed
techniques.

\end{itemize}

The remainder of this paper is organized as follows.  We first
motivate our approach in Section~\ref{sec:motivation} by comparing it
with testing, model checking, and abstract interpretation 
tools designed for threads.  Then, we provide the
technical background on interrupt modeling and abstract interpretation
in Section~\ref{sec:prelim}.  Next, we present our new method for
checking the feasibility of data flows between interrupts in
Section~\ref{sec:contraint}, followed by our method for integrating
the feasibility checking with abstract interpretation in
Section~\ref{sec:prior-analysis}.  We present our experimental
evaluation in Section~\ref{sec:experiment}.  Finally, we review the
related work in Section~\ref{sec:relatedwork} and conclude in
Section~\ref{sec:conclusion}.

\begin{figure}[!t]
\centering
\includegraphics[width=\linewidth]{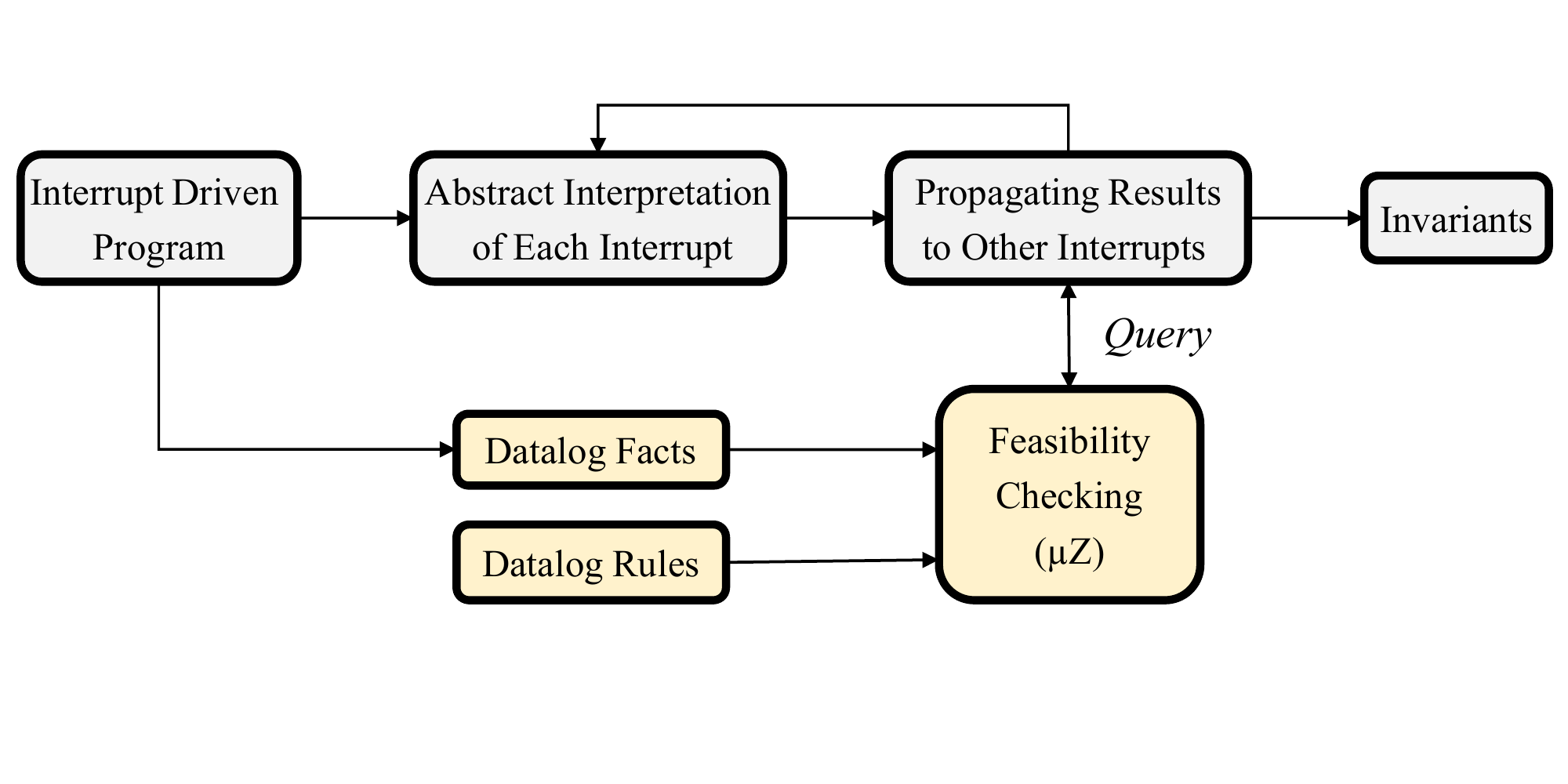}

\vspace{-6ex}
\caption{\Name{} -- iterative verification framework for interrupt-driven programs.}
\label{fig:overview}
\end{figure}

\section{Motivation}
\label{sec:motivation}

We first use examples to illustrate the problems of prior techniques
such as testing, model checking, and thread-modular abstract
interpretation.  Then, we explain how our method overcomes these
problems.

Consider the example program in Fig.~\ref{fig:mot0}, which has three
interrupts \texttt{irq\_H}, \texttt{irq\_L} and \texttt{irq\_M}.  The
suffix \texttt{H} is used to denote high priority, \texttt{L} for low
priority, and \texttt{M} for medium priority.  Interrupts with higher
priority levels may preempt interrupts with lower priority levels, but
not vice versa.  Inside the handlers, there are two
variables \texttt{x} and \texttt{y}, which are set to \texttt{0}
initially.  Among the three assertions, the first two may fail, while
the last one always holds.

\begin{figure}[!]
\vspace{1ex}
\centering
\hspace{0\linewidth}
  \centering
\begin{minipage}{\linewidth}
\centering
\framebox[.31\linewidth]{
\begin{minipage}{.31\linewidth}
{\scriptsize

~~\texttt{irq\_H() \{}

\medskip

~~~~~\texttt{...}

\medskip

~~~~~\texttt{assert(y==0);} 

~~\texttt{\}}
}
\end{minipage}
}
\framebox[.31\linewidth]{
\begin{minipage}{.31\linewidth}
{\scriptsize

~~\texttt{irq\_L() \{}

\medskip

~~~~~\texttt{x = 0;}

\medskip

~~~~~\texttt{assert(x==0);}

~~\texttt{\}}
}
\end{minipage}
}
\framebox[.31\linewidth]{
\begin{minipage}{.31\linewidth}
{\scriptsize

~~\texttt{irq\_M()} \{

\medskip

~~~~~\texttt{y = 1;}

~~~~~\texttt{x = 1;}

~~~~\texttt{assert(x==1);}

~~\}
}
\end{minipage}
}
\end{minipage}
\caption{An example program with three interrupt handlers and assertions.}
\label{fig:mot0}	
\end{figure}

\subsection{Testing}

Testing an interrupt-driven program requires the existence of
interrupt sequences, which must be generated a priori.  In
Fig.~\ref{fig:mot0}, for example, since the interrupt handlers have
different priority levels, we need to consider preemption while
creating the test sequences.  Since a high-priority interrupt handler
may preempt, at any time, the execution of a medium- or low-priority
interrupt handler, when \texttt{irq\_L} is executing, \texttt{irq\_H}
may be interleaved in between its instructions.

Fig.~\ref{fig:intSeq} shows four of the possible interrupt sequences
for the program.
Specifically, \emph{case0} is the sequential execution of the three
handler functions;
\emph{case1} shows that \texttt{irq\_H}
preempts \texttt{irq\_M}, followed by \texttt{irq\_L};
\emph{case2} is similar, except that \texttt{irq\_L} executes 
first and then is preempted by \texttt{irq\_H}, followed
by \texttt{irq\_M};  and 
\emph{case3} is the nested case where \texttt{irq\_L}
is preempted by \texttt{irq\_M} and then by \texttt{irq\_H}.

\begin{figure}
\centering
\includegraphics[width=0.675\linewidth]{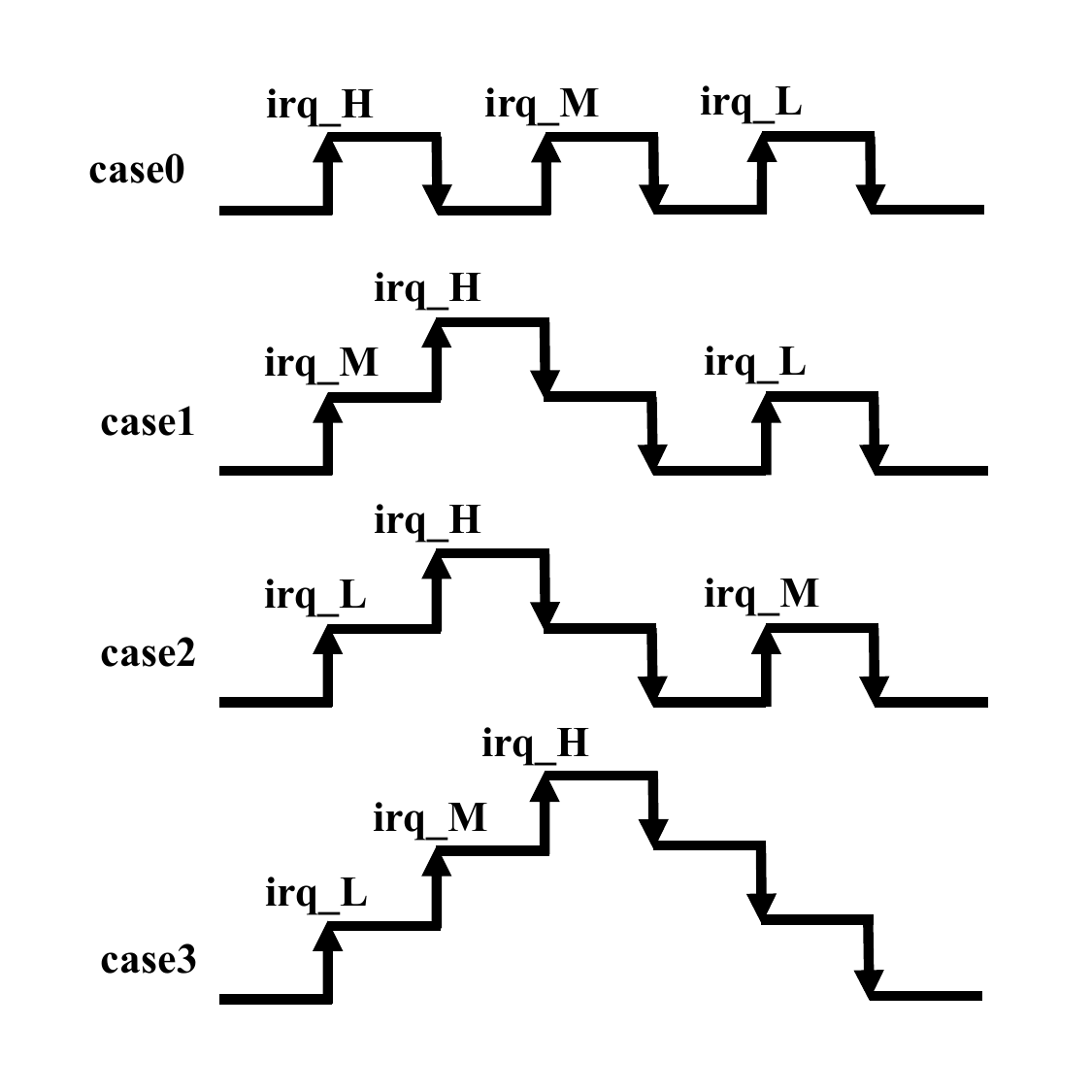}
\caption{Some possible interrupt sequences for the program in Fig.~\ref{fig:mot0}.}
\label{fig:intSeq}
\end{figure}

The main problem of testing is that there can be too many such interrupt
sequences to explore.  Even if we can somehow guarantee that each
interrupt handler is executed only once, the total number of test
sequences can be enormously large, even for small or medium-sized
programs.

\subsection{Model Checking}

Model checking tools such as CBMC~\cite{Clarke04b} may be used to search for erroneous
interrupt sequences, e.g., those leading to assertion violations.  For
instance, in the running example, all assertions hold under the
sequences \emph{case0} and \emph{case2} in Fig.~\ref{fig:intSeq}.
This is because, although \texttt{irq\_H} preempts \texttt{irq\_L},
they access different variables and thus do not affect the assertion
conditions, while \texttt{irq\_M} checks the value of \texttt{x} after
assigning 1 to \texttt{x}.

In \emph{case1}, however, the execution order of the three interrupt
handlers is different, thus leading to an assertion violation
inside \texttt{irq\_H}.  More specifically, \texttt{irq\_M} is
preempted by \texttt{irq\_H} at first.  Then, after both
completes, \texttt{irq\_L} is executed. So, the change of \texttt{y}
may affect the read of \texttt{y} in \texttt{irq\_H}, leading to the
violation.

Finally, in \emph{case3}, both of the first two assertions may be
violated, because the check of \texttt{x} in \texttt{irq\_L} and the
check of \texttt{y} in \texttt{irq\_H} can be affected by \texttt{irq\_M}'s own assignments 
of \texttt{x} and \texttt{y}.

Although bounded model checking can quickly find bugs, e.g., the assertion
violations in Fig.~\ref{fig:mot0}, the depth of the execution is often bounded, which means in practice,
tools such as iCBMC~\cite{Kroening15} cannot prove the absence of
bugs.

\subsection{Abstract Interpretation}

Abstract interpretation is a technique designed for
proving properties, e.g., assertions always hold.  Unfortunately,
existing methods based on abstract interpretation are mostly
designed for threads as opposed to interrupts.  Since threads interact
with each other more freely than interrupts, these methods are
essentially over-approximated analysis.  As such, they may still be
leveraged to prove properties in interrupt-driven programs, albeit in
a less accurate fashion.  That is, when they prove an assertion holds,
the assertion indeed holds; but when they cannot prove an assertion,
the result is inconclusive.

\begin{figure}
\centering
\includegraphics[width=0.75\linewidth]{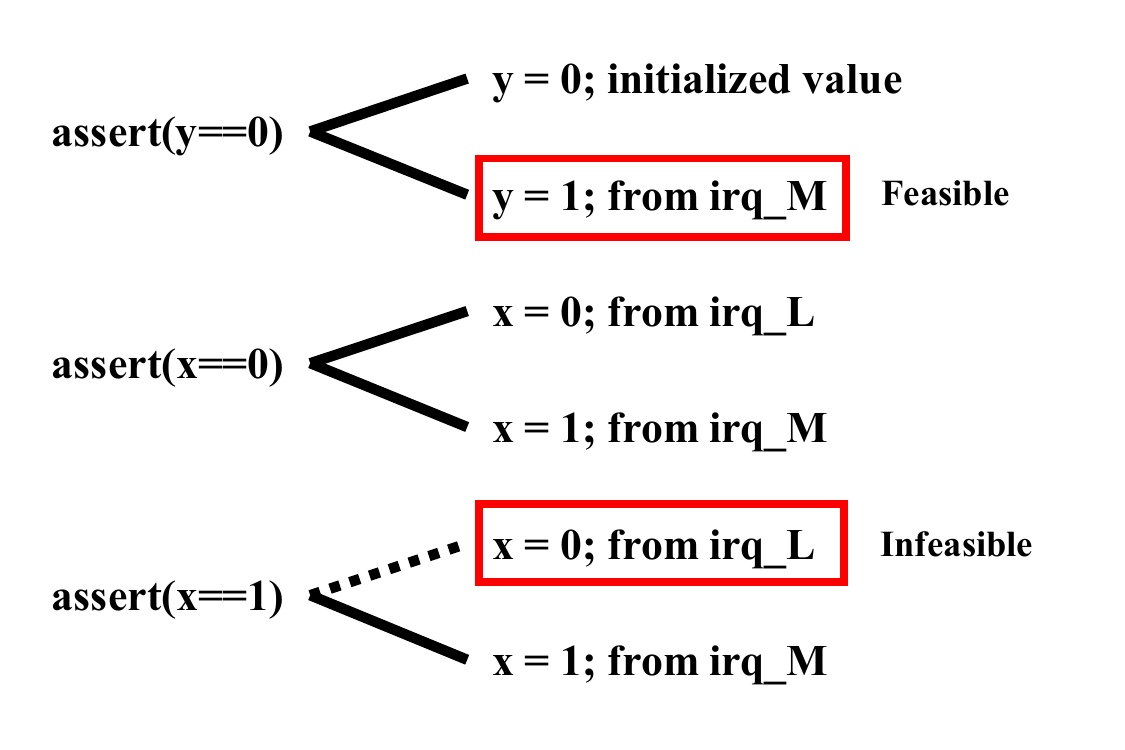}
\caption{Some possible \emph{store-to-load} data flows during abstract interpretation.}
\label{fig:ldstpairs}
\end{figure}

For the running example in Fig.~\ref{fig:mot0}, for instance, existing
abstract interpretation techniques such as
Min\'e~\cite{Mine14,Mine11}, designed for analyzing threads,
cannot prove any of the three assertions.  To see why, let us first
assume that interrupt handlers are thread routines.  During
thread-modular abstract interpretation, the verification procedure would first
gather all possible pairs of load and store instructions with respect
to the global variables, as shown in Fig.~\ref{fig:ldstpairs},
where each assertion has two possible loads.  Specifically, the load of \texttt{y}
in \texttt{irq\_H} corresponds to the initial value \texttt{0} and the
store in \texttt{irq\_M}.  The load of \texttt{x} in \texttt{irq\_L}
corresponds to the stores in \texttt{irq\_L} and \texttt{irq\_M}.  The
load of \texttt{x} in \texttt{irq\_M} corresponds to the stores
in \texttt{irq\_L} and \texttt{irq\_M}.

Since these existing methods~\cite{Mine14,Mine11} assume that all stores may
affect all loads, they would incorrectly report that all three
assertions may fail.  For example, it reports that the load of \texttt{x}
in \texttt{irq\_M} may (incorrectly) read from the store \texttt{x=0}
in \texttt{irq\_L} despite that \texttt{irq\_L} has a lower priority
and thus cannot preempt \texttt{irq\_M}.
In contrast, our new method can successfully prove the third
assertion.
Specifically, we model the behavior of interrupts with different
levels of priorities.  Due to the different priority levels,
certain \emph{store-to-load} data flows are no longer feasible, as
shown by the stores marked by red boxes in Fig.~\ref{fig:ldstpairs}:
these two stores have lower priority than the corresponding load in
the assertions.

\subsection{Abstract Interpretation for Interrupts}

Modeling the priority levels \emph{alone}, however, is not enough for
proving all assertions because even without preemption, a
low-priority interrupt may affect a high-priority interrupt.  Consider
the first red box in Fig.~\ref{fig:ldstpairs}.  Although \texttt{y=1}
from \texttt{irq\_M} cannot affect the load of \texttt{y}
in \texttt{irq\_H} through preemption, if \texttt{irq\_H} is invoked
after \texttt{irq\_M} ends, \texttt{y} can still get the
value \texttt{1}, thus leading to the assertion violation.  Therefore,
our new verification procedure has to consider all possible
sequential interleavings of the interrupt handlers as well.

Now, consider the program in Fig.~\ref{fig:mot1}, which has three
interrupt handlers \texttt{irq\_M}, \texttt{irq\_L}
and \texttt{irq\_H}.  In these handler functions, there are two global
variables \texttt{x} and \texttt{y}, which are set to \texttt{0}
initially.  Among the three assertions, the first two always hold,
whereas the last one may fail.  For ease of comprehension, we assume
the computer hardware running this program provides the \emph{sequentially consistent}
memory~\cite{Lamport78,ZhangKW15,KusanoW17}.
Note that \texttt{irq\_M} has two stores of \texttt{y}, one inside the
conditional branch and the other outside, and \texttt{irq\_H} has two
stores of \texttt{x}, one inside the conditional branch and the other
outside.

\begin{figure}[!]
\vspace{1ex}
\centering
\hspace{0\linewidth}
  \centering
\framebox[.9\linewidth]{
\begin{minipage}{.9\linewidth}
\centering
\begin{minipage}{.3\linewidth}
{\scriptsize

\texttt{irq\_M() \{}

\medskip

~~\texttt{if (...)}

~~~~~~\texttt{y = 0;}

\medskip

~~\texttt{y = 1;}

\medskip

~~\texttt{assert(x==1);}

\}
}
\end{minipage}
\begin{minipage}{.3\linewidth}
{\scriptsize

\texttt{irq\_L() \{}

\medskip

~~\texttt{...}

\medskip

~~\texttt{y = 1;}

~~\texttt{...}

\medskip

~~\texttt{assert(y==1);}

\}
}
\end{minipage}
\begin{minipage}{.3\linewidth}
{\scriptsize

\texttt{irq\_H() \{}

\medskip

~~\texttt{if (...)}

~~~~~~\texttt{x = 0;}

\medskip

~~\texttt{x = 1;}

\medskip

~~\texttt{assert(y==1);}

\}
}
\end{minipage}
\end{minipage}
}
\caption{An example program with three interrupt handlers, where the first two assertions alway hold but the last assertion may fail.}
\label{fig:mot1}	
\end{figure}

With prior thread-modular analysis~\cite{Mine11,Mine14,KusanoW16}, all
three assertions may fail because the store \texttt{y=0}
in \texttt{irq\_M} may be interleaved right before the assertions
in \texttt{irq\_L} and \texttt{irq\_H}.  Furthermore, the
store \texttt{x=0} in \texttt{irq\_H} executed before \texttt{irq\_M}
may lead to the violation of the assertion in \texttt{irq\_M}.
In contrast, with our precise modeling of the interrupt behavior, the
new method can prove that the first two assertions always hold.
Specifically, the assertion in \texttt{irq\_L} holds because, even if
it is preempted by \texttt{irq\_M}, the value of \texttt{y}
remains \texttt{1}.
Similarly, the assertion in \texttt{irq\_M} holds because, even if it
is preempted by \texttt{irq\_H}, the store \texttt{x=1}
post-dominates the store \texttt{x=0}, meaning the value
of \texttt{x} remains \texttt{1} after \texttt{irq\_H} returns.

\begin{table*}[ht!]
\caption{Comparing \Name{} with Testing and Prior Verification Methods on the Programs in Fig.~\ref{fig:mot0} and Fig.~\ref{fig:mot1}.}
\label{tbl:difference}
\centering
\footnotesize
\scalebox{0.9}{
\begin{tabular}{|l||c|c||c|c|}\hline
 Property
             & ~~~~~Testing~\cite{Regehr05,Higashi10}~~~~~
             & Model Checking (bounded)~\cite{Kroening15,WuCMDW16} 
             & Abs.\ Int. for Threads~\cite{Mine14,KusanoW16}  
             & \Name{} for Interrupts (new)  \\
             \hline\hline

assertion in Fig.~\ref{fig:mot0}: \texttt{irq\_H}~~~~~ 
             &   violation   &   violation   &    warning         &  warning       \\\hline
assertion in Fig.~\ref{fig:mot0}: \texttt{irq\_L} 
             &   violation   &   violation   &    warning         &  warning       \\\hline
assertion in Fig.~\ref{fig:mot0}: \texttt{irq\_M} 
             &               &               &    (bogus) warning   &   \textbf{proof}        \\\hline

assertion in Fig.~\ref{fig:mot1}: \texttt{irq\_M} 
             &               &               &    (bogus) warning   &   \textbf{proof}        \\\hline
assertion in Fig.~\ref{fig:mot1}: \texttt{irq\_L} 
             &               &               &    (bogus) warning   &   \textbf{proof}        \\\hline
assertion in Fig.~\ref{fig:mot1}: \texttt{irq\_H} 
             &  violation    &   violation   &    warning         &  warning       \\\hline

\end{tabular}
}
\end{table*}

In contrast, the assertion in \texttt{irq\_H} may fail
if \texttt{irq\_H} preempts \texttt{irq\_M} right after the
conditional branch that sets \texttt{y} to \texttt{0}.  This
particular preemption is feasible because \texttt{irq\_H} has a higher
priority than \texttt{irq\_M}.

Therefore, our new method has to consider not only the different
levels of priority of all interrupts, but also the domination and
post-domination relations within each handler. It decides the
feasibility of \emph{store-to-load} data flows based on whether a load
has a dominated store, whether a store has a post-dominated store, and
whether a load-store pair is allowed by the priority levels of the
interrupts.  We present the details of this
\emph{feasibility-checking} algorithm in
Section~\ref{sec:contraint}.

To sum up, the main advantages of \Name{} over state-of-the-art
techniques are shown in Table~\ref{tbl:difference}.  Specifically,
testing and (bounded) model checking tools are good at detecting bugs
(e.g., assertion violations) but cannot prove the absence of bugs,
whereas thread-modular abstract interpretation tools are good at
obtaining proofs, but may report many false positives (i.e., bogus warnings).
In contrast, our new \emph{abstract interpretation} method is
significantly more accurate.  It can obtain more proofs than prior
techniques and, at the same time, can significantly reduce the number of bogus warnings.

\section{Preliminaries}
\label{sec:prelim}

In this section, we describe how interrupt-driven programs are
modeled in our framework by comparing their behavior to the
behavior of threads.  Then, we review the basics of prior abstract
interpretation techniques.

\subsection{Modeling of Interrupts}

We consider an interrupt-driven program as a finite set $T =\{
T_1,\ldots,T_n\}$ of sequential programs.  Each sequential program
$T_i$, where $1\leq i\leq n$, denotes an interrupt handler.  For ease
of presentation, we do not distinguish between the main program and
the interrupt handlers.  Globally, sequential programs in $T$ are
executed in a strictly interleaved fashion.  Each sequential program
may access its own local variables; in addition, it may access a set
of global variables, through which it communicates with the other
sequential programs in $T$.

The interleaving behavior of interrupts is a strict subset of the
interleaving behavior of threads (c.f.~\cite{Kroening15}).  This is
because concurrently running threads are allowed to freely preempt
each other's executions.  However, this is not the case for
interrupts.

\begin{figure}
\vspace{2ex}
\centering

\begin{minipage}{.8\linewidth}

\framebox[.35\linewidth]{
\begin{minipage}{.35\linewidth}

{\scriptsize

~~~~~\texttt{run0() \{}

\medskip

~~~~~~~~\texttt{stmt1;}

~~~~~~~~\texttt{stmt2;}

~~~~~\texttt{\}}
}

\end{minipage}
}
\framebox[.35\linewidth]{
\begin{minipage}{.35\linewidth}
{\scriptsize

~~~~~\texttt{run1() \{}

\medskip

~~~~~~~~\texttt{stmt3;}

~~~~~~~~\texttt{stmt4;}

~~~~~\texttt{\}}
}

\end{minipage}
}

\end{minipage}

\vspace{1ex}
\begin{minipage}{.85\linewidth}

{\scriptsize

\begin{itemize}

\item Possible traces for interrupts:
\begin{itemize}
\item \texttt{stmt1 $\rightarrow$ stmt2 $\rightarrow$ stmt3 $\rightarrow$ stmt4}
\item \texttt{stmt1 $\rightarrow$ stmt3 $\rightarrow$ stmt4 $\rightarrow$ stmt2}
\end{itemize}  

\item Possible traces for threads: 
\begin{itemize}
\item \texttt{stmt1 $\rightarrow$ stmt2 $\rightarrow$ stmt3 $\rightarrow$ stmt4}
\item \texttt{stmt1 $\rightarrow$ stmt3 $\rightarrow$ stmt4 $\rightarrow$ stmt2}
\item \texttt{stmt1 $\rightarrow$ stmt3 $\rightarrow$ stmt2 $\rightarrow$ stmt4}
\end{itemize}  

\end{itemize}

}
\end{minipage}


\vspace{1ex}
\caption{The interleavings (after \texttt{stmt1}) allowed by interrupts and threads.}
\label{fig:compBehavior}	
\end{figure}

Consider the example program in Fig.~\ref{fig:compBehavior}, which has
two functions named \texttt{run0} and \texttt{run1}.
If they were interrupts, where \texttt{run1} has a higher priority
level than \texttt{run0}, then after executing \texttt{stmt1}, there
can only be two possible traces.  The first one is for \texttt{run1}
to wait until \texttt{run0} ends, and the second one is
for \texttt{run1} to preempt \texttt{run0}.
If they were threads, however, there can be three possible traces
after executing \texttt{stmt1}.  In addition to the traces allowed by
interrupts, we can also execute \texttt{stmt3} in \texttt{run1}, then
execute \texttt{stmt2} in \texttt{run0}, and finally
execute \texttt{stmt4} in \texttt{run1}.  
The third trace is infeasible for interrupts because the
high-priority \texttt{run1} cannot be preempted by \texttt{stmt2} of
the low-priority \texttt{run0}.

Since the interleaving behavior of interrupts is a strict subset of
the interleaving behavior of threads, it is always safe to apply a
sound static verification procedure designed for threads to interrupts.  If the
verifier can prove the absence of bugs by treating interrupts as
threads, then the proof is guaranteed to be valid for interrupts.  The
result of this discussion can be summarized as follows:

\vspace{1ex}
\begin{theorem}
Since the interleaving behavior of interrupts is a subset of the interleaving
behavior of threads, proofs obtained by any sound abstract interpretation over
threads remain valid for interrupts.
\end{theorem}
\vspace{1ex}

However, the reverse is not true:  a bug reported by the
verifier that treats interrupts as threads may not be a real bug since the
erroneous interleaving may be infeasible.  In practice, there are
also tricky corner cases during the interaction of interrupts, such as
nested invocations of handlers, which call for a more
accurate modeling framework for interrupts.

Interrupts may also be invoked in a nested fashion as shown
by \emph{case3} in Fig.~\ref{fig:intSeq}, which complicates the static
analysis.  Here, we say interrupts are nested when one's handler
function is invoked before another's handler function returns, and
the third handler function is invoked before the second handler
function returns.  Such nested invocations are possible, for example, if
the corresponding interrupts have different priority levels, where the
inner most interrupt has the highest priority level.  This 
behavior is different from thread interleaving; with numerous corner
cases, it requires the development of dedicated modeling and
analysis techniques.

\subsection{Abstract Interpretation for Threads}
\label{sec:thread-mod}

Existing methods for modular abstract interpretation are designed
almost exclusively for multithreaded
programs~\cite{Mine11,Mine12,Mine14,KusanoW16}.  Typically, the
analyzer works on each thread in isolation, without creating a
monolithic verification model as in the non-modular 
techniques~\cite{FarzanK12,WuCMDW16}, to avoid the
up-front complexity blowup.  
At a high level, the analyzer iterates through threads in two steps: (1)
analyzing each thread in isolation, and (2) propagating results from
the shared-memory writes of one thread to the corresponding reads of
other threads.

Let the entire program $P$ be a finite set of threads, where each thread
$T$ is represented by a control-flow graph $\langle N, n_0, \delta \rangle$ with a
set of nodes $N$, an entry node $n_0$, and the transition relation
$\delta$.  Each pair $(n,n') \in \delta$ means control may flow from
$n$ to $n'$.
Each node $n$ is associated with an abstract memory-state
over-approximating the possible concrete states at $n$.
We assume the abstract domain (e.g., intervals) is defined as a
lattice with appropriate top ($\top$) and bottom ($\bot$) elements, a
partial-order relation ($\sqsubseteq$), and widening/narrowing
operators to ensure that the analysis eventually terminates~\cite{Cousot77}. 
We also define an interference $I$ that maps a variable $v$ to the
values stored into $v$ by some thread $T$.

Algorithm~\ref{alg:analyze-thread} shows how a thread-local analyzer works
on $T$ assuming some interferences $I$ provided by the environment
(e.g., writes in other threads). It treats $T$ as a sequential program.
Let $S(n)$ be the abstract memory-state at node $n$, $n_0$ be the
entry node of $T$, and $W$ be the set of nodes in $T$ left to be
processed.
The procedure keeps removing node $n$ from the work-list $W$ and
processing it until $W$ is empty (i.e., a fixed point is reached).

If node $n$ corresponds to a shared-memory read of variable $v$, then the transfer function $\tfunc$ (Line~7) assumes that
$n$ can read either the local value (from $S$) or the value written by
another thread (the interference $I(v)$).
The transfer function $\tfunc$ of an instruction $n$ takes some memory-state as input and returns a new memory-state as output; the new memory-state is the result of executing the instruction in the given memory-state. 
Otherwise, if $n$ is a local read, in which case the transfer function $\tfunc$
uses the local memory-state (Line~9) as in the abstract interpretation of any
sequential program.
The analysis result (denoted $S$) is an over-approximation of the memory
states within $T$ assuming interferences $I$.

\begin{algorithm}[t!]
  \caption{Local analysis of $T$ with prior interferences $I$.}
  \label{alg:analyze-thread}
{\footnotesize
\begin{algorithmic}[1]
  \Function{AnalyzeLocal}{$T = \langle N, n_0, \delta \rangle, I$}
    \State $S \gets \varnothing$ \Comment{Map from nodes to states}
    \State $W \gets \{n_0\}$ \Comment{Set of nodes to process}
    \While{$\exists n \in W$}
      \State $W \gets W \setminus \{n\}$
      \If {$n$ is a shared-memory read of variable $v$}
        \State $s \gets \tfunc(n, S(n) \sqcup I(v))$
      \Else
        \State $s \gets \tfunc(n, S(n))$
      \EndIf
      \ForAll{$\langle n, n' \rangle \in \delta$ such that $s \not\sqsubseteq
          S(n')$}
        \State $S(n') \gets S(n') \sqcup s$
        \State $W \gets W \cup \{n'\}$
      \EndFor

    \EndWhile
    \State \Return $S$
  \EndFunction
\end{algorithmic}
}
\end{algorithm}

The procedure that analyzes the entire program is shown in
Algorithm~\ref{alg:analyze-prog}.
It first analyzes each thread, computes the interferences, and then analyzes each
thread again in the presence of these interferences.
The iterative process continues until a fixed point on the memory-states of all threads is reached.
Initially, $S$ maps each node in the program to an empty memory-state
$\bot$.  $S'$ contains the analysis results after one iteration
of the fixed-point computation.
The function \textsc{Interf} returns the interferences of thread $T$,
i.e., a map from some variable $v$ to all the (abstract) values stored into $v$
by $T$.
Each thread $T$ is analyzed in isolation by the loop at Lines 4--9.
Here, we use $\uplus$ to denote  the join ($\sqcup$) of all  memory-states on the matching nodes.

\begin{algorithm}[t!]
  \caption{Analysis of the entire program, i.e., a set of $T$'s.}
  \label{alg:analyze-prog}
{\footnotesize
\begin{algorithmic}[1]
  \Function{AnalyzeProg}{$P$}
    \State $S \gets \text{map all nodes to $\bot$}$
    \State $S' \gets S$
    \Repeat
      \State $S = S'$
      \ForAll{$T \in P$}
        \State $I \gets \biguplus \Call{Interf}{T', S}$ for each $T' \in P, T' \neq T$
        \State $S' \gets S' \uplus \Call{AnalyzeLocal}{T, I}$
      \EndFor
    \Until{$S' = S$}
  \EndFunction
  \Function{Interf}{$T = \langle N, n_0, \delta \rangle$, S}
    \State $I \gets \varnothing$
    \ForAll{$n \in N$}
      \If {$n$ is a shared memory write to variable $v$}
        \State $I(v) \gets I(v) \sqcup \tfunc(n, \mathit{S}(n))$
      \EndIf
    \EndFor
    \State \Return $I$
  \EndFunction
\end{algorithmic}
}
\end{algorithm}

This \emph{thread-modular} abstract interpretation framework, while more efficient
than monolithic verification, is potentially less accurate.
For example, a load $l$ may see \emph{any} value written into the
shared memory by a store $s$ even if there does not exist a path in
the program where $l$ observes $s$.
This is why, as shown in Table~\ref{tbl:difference}, techniques
such as ~\cite{Mine14,KusanoW16} cannot obtain proofs for the 
programs in Fig.~\ref{fig:mot0} and Fig.~\ref{fig:mot1}.
In the context of interrupt handlers with priorities, it means that
even infeasible store-to-load flows due to priorities may 
be included in the  analysis, thus causing
false alarms.  

In the remainder of this paper, we show how to introduce priorities
into the propagation of data flows between interrupts during the
analysis, thereby increasing the accuracy while retaining its 
efficiency.

\section{Feasibility of Data Flows between Interrupts}
\label{sec:contraint}

In this section, we present our method for precisely modeling the
priority-based interleaving semantics of interrupts, and deciding the
feasibility of \emph{store-to-load} data flows between interrupts.
If, for example, a certain \emph{store-to-load} data flow is indeed
not feasible, it will not be propagated across interrupts in
Algorithm~\ref{alg:analyze-prog}.

More formally, given a set of \emph{store-to-load} pairs, we want to
compute a new \textrel{MustNotReadFrom} relation, such that
$\textrel{MustNotReadFrom}(l,s)$, for any load $l$ and store $s$,
means if we respect all the other existing \emph{store-to-load} pairs, then it
would be infeasible for $l$ to get the value written by $s$.

We have developed a Datalog-based declarative program analysis procedure for
computing \textrel{MustNotReadFrom}.  Toward this end, we first
generate a set of \emph{Datalog facts} from the program and the
given \emph{store-to-load} pairs.  Then, we generate a set
of \emph{Datalog rules}, which infer the new \textrel{MustNotReadFrom}
relation from the Datalog facts.  Finally, we feed the facts together
with the rules to an off-the-shelf Datalog engine, which computes 
the \textrel{MustNotReadFrom} relation.  In our implementation, we used the $\mu$-Z Datalog engine~\cite{Hoder11}
to solve the Datalog constraints.

\subsection{Inference Rules}

Before presenting the rules, we define some  relations:
\begin{compactitem}
\item $\textrel{Dom}(a,b)$: statement $a$ dominates $b$ in the CFG of an interrupt handler function. 
\item $\textrel{PostDom}(a,b)$: statement $a$ post-dominates $b$ in the CFG of an interrupt handler function. 
\item $\textrel{Pri}(s,p)$: statement $s$ has the priority level $p$.
\item $\textrel{Load}(l, v)$:  $l$ is a load of global variable $v$.
\item $\textrel{Store}(s, v)$:  $s$ is a store to global variable $v$.
\end{compactitem}
Dominance and post-dominance are efficiently
computable~\cite{Ferrante87} within each interrupt handler (not across
interrupt handlers).
Priority information for each interrupt handler, and thus all its statements,
may be obtained directly from the program. 
Similarly, \textrel{Load} and \textrel{Store} relations may be directly obtained from the
program.

Next, we present the rules for inferring three new relations: \textrel{NoPreempt}, \textrel{CoveredLoad} and \textrel{InterceptedStore}.

\vspace{1ex}
\paragraph{\textrel{NoPreempt}}

The relation means $s_1$ cannot preempt $s_2$, where $s_1$ and
$s_2$ are instructions in separate interrupt handlers.
From the interleaving semantics of interrupts, we know a handler may
only be preempted by another handler with a higher priority.
Thus,

\begin{equation*}
{\footnotesize
  \textrel{NoPreempt}(s_1, s_2)
    \leftarrow \textrel{Pri}(s_1, p_1)
               \land \textrel{Pri}(s_2, p_2)
               \land (p_2 \geq p_1)
}
\end{equation*}

\noindent
Here, $\textrel{Pri}(s_1,p_1)$ means $s_1$ belongs to a handler with
priority $p_1$, and $\textrel{Pri}(s_2,p_2)$ means $s_2$ belongs to a
handler with priority $p_2$. If $p_1$ is not higher
than $p_2$, then $s_1$ cannot preempt $s_2$.

\vspace{1ex}
\paragraph{\textrel{CoveredLoad}}

The relation means a load $l$ of a variable $v$ is covered by
a store $s$ to $v$ inside the same interrupt handler; this is the case when $s$
occurs before $l$ along all program paths.  This is captured
by the \emph{dominance} relation in the corresponding control flow graph:

\begin{equation*}
{\footnotesize
  \textrel{CoveredLoad}(l) \leftarrow
    \textrel{Load}(l, v)
    \land \textrel{Store}(s, v)
    \land \textrel{Dom}(s, l)
}
\end{equation*}

\vspace{1ex}
\paragraph{\textrel{InterceptedStore}}

The relation is similar to \textrel{CoveredLoad}.  We say a store
$s_1$ is intercepted by another store $s_2$ if $s_2$ occurs after $s_1$ along all program paths in the same handler.
Intuitively, the value written by $s_1$ is always overwritten by $s_2$ before the
handler terminates.  Formally,

\begin{equation*}
{\footnotesize%
\begin{aligned}%
  \textrel{InterceptedStore}(s_1) \leftarrow
    \textrel{Store}(s_1, v) 
    \land \textrel{Store}(s_2, v) \\
    \phantom{e} \land \textrel{PostDom}(s_2, s_1)
\end{aligned}%
}%
\end{equation*}

Finally, the \textrel{MustNotReadFrom} relation is deduced using all
aforementioned relations including \textrel{NoPreempt},  \textrel{CoveredLoad}
and \textrel{InterceptedStore}.  It indicates that, under the current
situation (defined by the set of existing store-to-load data flows), a
load $l$ cannot read from a store $s$ in any feasible interleaving.
There are several cases:

First, we say a load $l$ covered by a store in a handler $I$ cannot 
read from a store $s$ intercepted by another store in a handler $I'$,
 because $l$ cannot read from $s$ using any preemption or by running $I$ and $I'$ sequentially. 

\begin{equation*}
{\footnotesize%
\begin{aligned}%
  \textrel{MustNotReadFrom}(l, s) \leftarrow
    \textrel{CoveredLoad}(l)
    \land \textrel{Load}(l,v) \\
    \phantom{e} \land \textrel{Store}(s,v)
    \land \textrel{InterceptedStore}(s)
\end{aligned}%
}%
\end{equation*}

Second, we say a load $l$ covered by a store $s$ in a handler $I$
cannot read from any store $s'$ that cannot preempt $I$, because
the value of $s'$ will always be overwritten by $s$. 
That is, since $s'$ cannot preempt $I$, it cannot execute in between $s$ and $l$. 

\begin{equation*}
{\footnotesize%
\begin{aligned}%
  \textrel{MustNotReadFrom}(l, s) \leftarrow
    \textrel{CoveredLoad}(l)
    \land \textrel{Load}(l,v) \\
    \phantom{e} \land \textrel{Store}(s,v)
    \land \textrel{NoPreempt}(s, l)
\end{aligned}%
}%
\end{equation*}

Third, we say that, if a store $s$ is intercepted in a handler $I$,
then a load $l$ of the same variable that cannot preempt $s$, cannot 
read from the value stored by $s$. This is because the store
 intercepting $s$ will always overwrite the value. 

\begin{equation*}
{\footnotesize%
\begin{aligned}%
  \textrel{MustNotReadFrom}(l, s) \leftarrow
    \textrel{InterceptedStore}(s)
    \land \textrel{Store}(s,v) \\
    \phantom{e} \land \textrel{Load}(l,v)
    \land \textrel{NoPreempt}(l, s)
\end{aligned}%
}%
\end{equation*}

\subsection{The Running Examples}

To help understand how \textrel{MustNotReadFrom} is deduced from the
 Datalog rules and facts, we provide a few examples.
For ease of comprehension, we show in Table~\ref{tbl:rules} how \textrel{MustNotReadFrom} may be 
deduced from \textrel{InterceptedStore}, \textrel{CoveredLoad}
and \textrel{NoPreempt}.
Since all stores are either in or outside \textrel{InterceptedStore},
and all loads are either in or outside \textrel{CoveredLoad}, our
rules capture the \textrel{MustNotReadFrom} relation between all
stores and loads.

\begin{table}[ht]
\caption{MustNotReadFrom rules based on InterceptedStore, CoveredLoad and Priority}
\label{tbl:rules}
\centering
\resizebox{\linewidth}{!}{
\begin{tabular}{l|l|l}
\toprule
& \textrel{InterceptedStore}($s$) & \textrel{Not InterceptedStore}($s$)  \\
\midrule

\textrel{CoveredLoad}($l$)& 
\textrel{MustNotReadFrom}($l$, $s$) &
1) Not \textrel{NoPreempt}($s$, $l$) $\rightarrow $ \\
&& Possibly \textrel{ReadFrom}($l$, $s$)   \\
&& 2) \textrel{NoPreempt}($s$, $l$) $\rightarrow $ \\
&& \textrel{MustNotReadFrom}($l$, $s$)
\\
\midrule
\textrel{Not CoveredLoad}(l)&
1) \textrel{NoPreempt}($l$, $s$) $\rightarrow $ &  
Possibly \textrel{ReadFrom}($l$, $s$) \\
& \textrel{MustNotReadFrom}($l$, $s$) &  \\
& 2) Not \textrel{NoPreempt}($l$, $s$) $\rightarrow $ & \\
&  Possibly \textrel{ReadFrom}($l$, $s$) &\\

\bottomrule
\end{tabular}
}
\end{table}

Specifically, if a store $s$ is \textrel{InterceptedStore} and a load
$l$ is \textrel{CoveredLoad}, there is no way for the load to read
from the store (Row 2 and Column 2).

If a load $l$ is not \textrel{CoveredLoad} and a store $s$ is
not \textrel{InterceptedStore}, the load may read from the store by
running sequentially or via preemption (Row 3 and Column 3).

If a load $l$ is \textrel{CoveredLoad}, a store $s$ is
not \textrel{InterceptedStore} and the handler of the store can
preempt the handler of the load, the load may read from the store through
preemption (the first case at Row 2 and Column 3).  However, if the
handler of the store cannot preempt the handler of the load, it is
impossible for the load to read from the store; in this case, the load always reads
from the store in the same interrupt handler (the second case at Row 2 and
Column 3).

Lastly, if a load $l$ is not \textrel{CoveredLoad}, a store $s$
is \textrel{InterceptedStore}, and the handler of the load cannot
preempt the handler of the store, the load cannot read from the store
since the value of the store is always overwritten by another store in
the same handler (the first case at Row 3 and Column 2).  However, if the
handler of the load can preempt the handler of the store, then the
load can read from the store through preemption in between two stores (the second
case at Row 3 and Column 2).

\begin{figure}[ht]
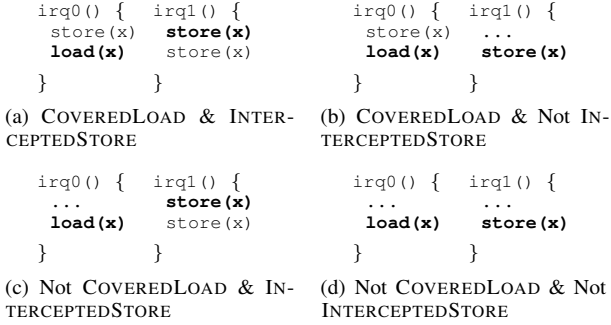

\centering
\subfigure[\textrel{CoveredLoad} \& \textrel{InterceptedStore}]{
\centering
\begin{minipage}{.4\linewidth}
\centering
\begin{minipage}{.4\linewidth}
{\scriptsize

\texttt{irq0() \{}

~~\texttt{store(x)}

~~\texttt{\textbf{load(x)}}

\}
\vspace{2ex}
}
\end{minipage}
\centering
\begin{minipage}{.4\linewidth}
{\scriptsize

\texttt{irq1() \{}

~~\texttt{\textbf{store(x)}}

~~\texttt{store(x)}

\}
\vspace{2ex}
}
\end{minipage}
\end{minipage}
\label{fig:case1}
}
\hspace{1ex}
\subfigure[\textrel{CoveredLoad} \& Not \textrel{InterceptedStore}]{
\centering
\begin{minipage}{.4\linewidth}
\centering
\begin{minipage}{.4\linewidth}
{\scriptsize

\texttt{irq0() \{}

~~\texttt{store(x)}

~~\texttt{\textbf{load(x)}}

\}
\vspace{2ex}
}
\end{minipage}
\centering
\begin{minipage}{.4\linewidth}
{\scriptsize

\texttt{irq1() \{}

~~\texttt{...}

~~\texttt{\textbf{store(x)}}

\}
\vspace{2ex}
}
\end{minipage}
\end{minipage}
\label{fig:case2}
}
\vspace{2ex}
\subfigure[Not \textrel{CoveredLoad} \& \textrel{InterceptedStore}]{
\centering
\begin{minipage}{.4\linewidth}
\centering
\begin{minipage}{.4\linewidth}
{\scriptsize

\texttt{irq0() \{}

~~\texttt{...}

~~\texttt{\textbf{load(x)}}

\}
\vspace{2ex}
}
\end{minipage}
\centering
\begin{minipage}{.4\linewidth}
{\scriptsize

\texttt{irq1() \{}

~~\texttt{\textbf{store(x)}}

~~\texttt{store(x)}

\}
\vspace{2ex}
}
\end{minipage}
\end{minipage}
\label{fig:case3}
}
\hspace{1ex}
\subfigure[Not \textrel{CoveredLoad} \& Not \textrel{InterceptedStore}]{
\centering
\begin{minipage}{.4\linewidth}
\centering
\begin{minipage}{.4\linewidth}
{\scriptsize

\texttt{irq0() \{}

~~\texttt{...}

~~\texttt{\textbf{load(x)}}

\}
\vspace{2ex}
}
\end{minipage}
\centering
\begin{minipage}{.4\linewidth}
{\scriptsize

\texttt{irq1() \{}

~~\texttt{...}

~~\texttt{\textbf{store(x)}}

\}
\vspace{2ex}
}
\end{minipage}
\end{minipage}
\label{fig:case4}
}
\caption{Examples for each case in Table~\ref{tbl:rules}.}
\label{fig:caseExamples}
\end{figure}

Fig.~\ref{fig:caseExamples} shows concrete examples of the four cases presented  in 
Table~\ref{tbl:rules}, where the figures correspond to the cases. 

Fig.~\ref{fig:case1} represents the case at Row 2 and Column 2 in Table~\ref{tbl:rules} and
Fig.~\ref{fig:case2} represents the case at Row 2 and Column 3.  In both programs, only the interference between
bold-style statements are considered.

Fig.~\ref{fig:case1} shows an interference
between \textrel{CoveredLoad} and \textrel{InterceptedStore}.  Since
the load in \texttt{irq0} is always overwritten by another store in it and a
value of the first store in \texttt{irq1} is always updated by a store
in it, the load cannot read a value from the store by preemption
or running sequentially.

Fig.~\ref{fig:case2} shows an interference
between \textrel{CoveredLoad} and not \textrel{InterceptedStore}.  In
this case, if \texttt{irq1} can preempt \texttt{irq0}, then the store
from \texttt{irq1} can occur between the store and the load
in \texttt{irq0}.  Otherwise, the load from \texttt{irq0} cannot read
a value from the store from \texttt{irq1}.

Fig.~\ref{fig:case3} shows an interference between
not \textrel{CoveredLoad} and \textrel{InterceptedStore}.  Similarly,
if \texttt{irq0} can preempt \texttt{irq1}, the load
from \texttt{irq0} can occur between the two stores from \texttt{irq1}.
Thus, it is possible for the load to read a value from the first store
in \texttt{irq1}.  Otherwise, the load cannot read a value from the
store by preemption of \texttt{irq1} or running sequentially.

Fig.~\ref{fig:case4} shows an interference between
not \textrel{CoveredLoad} and not \textrel{InterceptedStore}.  Here,
the load in \texttt{irq0} can read a value from the store
in \texttt{irq1} by running sequentially or through preemption.  Therefore,
it is possible for the load to read a value from the store 
 as described at Row 3 and Column 3 in Table~\ref{tbl:rules}. 

To sum up, we can use the three inference rules to determine
  infeasible store-to-load pairs for all these cases.

\subsection{Soundness of the Analysis}

By soundness, we mean the \textrel{MustNotReadFrom} relation deduced
from the Datalog facts and rules is an \emph{under-approximation}.
That is, any pair $(l,s)$ of load and store in this relation is
guaranteed to be infeasible.  However, we do not attempt to identify
\emph{all} the infeasible pairs because the goal here is to \emph{quickly} identify 
\emph{some} infeasible pairs and skip them during the more expensive
abstract interpretation computation.

\vspace{1ex}
\begin{theorem}
Whenever $\textrel{MustNotReadFrom}(l,s)$ holds, the load $l$ cannot
  read from the store $s$ on any concrete execution of the 
  program.  \label{thrm:mnrf}
\end{theorem}
\vspace{1ex}

The soundness of our analysis as stated above can be established in two steps.  First, assume that 
each individual rule is correct, the composition is also correct.  Second, while presenting these rules, we have
sketched the intuition behind the correctness of each rule.
A more rigorous proof can be formulated via proof-by-contradiction in
a straightforward fashion, which we omit for brevity.

\section{The Overall Analysis Procedure}
\label{sec:prior-analysis}

We now explain how to integrate the feasibility checking technique
into the overall procedure for iterative analysis, which leverages
the \textrel{MustNotReadFrom} relation to improve performance.
Specifically, when analyzing each interrupt handler $T$, we filter out
any interfering stores from other interrupt handlers that are deemed
infeasible, thereby preventing their visibility to $T$.
This can be implemented in Algorithm~\ref{alg:analyze-prog} by
modifying the function \proc{Interf}, as well as the function
\proc{AnalyzeLocal} defined in Algorithm~\ref{alg:analyze-thread}.

Our modifications to \textproc{Interf} are shown in
Algorithm~\ref{alg:interfs-mod}.  That is, when computing the
interferences of $T$, we choose to create a set of store--state pairs,
instead of eagerly joining all these states.
By delaying the join of these states, we obtain the opportunity to
filter out the infeasible store-to-load pair individually. 
For this reason, we overload the definition of $\uplus$ to be
the join ($\sqcup$) of sets on matching variables.

\begin{algorithm}[t!]
  \caption{Analysis of the entire program (cf.\
  Alg.~\ref{alg:analyze-prog}).}
  \label{alg:interfs-mod}
{\footnotesize
\begin{algorithmic}[1]
\makeatletter
\setcounter{ALG@line}{9}
\makeatother
  \Function{Interf}{$T = \langle N, n_0, \delta \rangle$, $T' = \langle N', n'_0, \delta' \rangle$,  S}
    \State $I \gets \varnothing$
    \ForAll{$n \in N$}
      \If {$n$ is a shared memory write to variable $v$}
        \State $I(v) \gets I(v) \uplus \{(n, \tfunc(n, \mathit{S}(n))\}$
      \EndIf
    \EndFor
    \State \Return $I$
  \EndFunction
\end{algorithmic}
}
\end{algorithm}

Next, we modify the abstract interpretation procedure for a single
interrupt handler as shown in
Algorithm~\ref{alg:analyze-thread-mod}.  The process remains the same
as \proc{AnalyzeLocal} of Algorithm~\ref{alg:analyze-thread}
except that, when a load $l$ is encountered (Line 6), we join the
state from all interfering stores while removing any that must not
interfere with $l$, as determined by the \textrel{MustNotReadFrom}
relation (Line 7).

\begin{algorithm}[t!]
  \caption{Analysis of a single interrupt (cf.\
  Alg.~\ref{alg:analyze-thread}).}
  \label{alg:analyze-thread-mod}
{\footnotesize%
\begin{algorithmic}[1]
  \Function{AnalyzeLocal}{$T = \langle N, n_0, \delta \rangle, I$}
  \Statex \ldots
\makeatletter
\setcounter{ALG@line}{5}
\makeatother
      \If {$n$ is a shared-memory read of variable $v$}
        \State $i \gets \bigsqcup\{n \mid (\mathit{st},s) \in I(v) 
            \land \lnot\textrel{MustNotReadFrom}(l,\mathit{st})\}$
        \State $s \gets \tfunc(n, S(n) \sqcup i)$
      \Else
        \State $s \gets \tfunc(n, S(n))$
      \EndIf
      \Statex \ldots

  \EndFunction
\end{algorithmic}%
}
\end{algorithm}

The remainder of the modular analysis remains the same as in
Algorithm~\ref{alg:analyze-prog}.

For example, in Fig.~\ref{fig:case1}, existing thread-modular abstract
interpretation methods would consider the two stores from \texttt{irq1}
for the load of \texttt{x} in \texttt{irq0}.  In contrast, we use
Algorithm~\ref{alg:analyze-thread-mod} to remove the pairing of the load
in \texttt{irq0} and the first store in \texttt{irq1}, since the load
and the store satisfy the \textrel{MustNotReadFrom} relation.
Similarly, the pairing of the load of \texttt{x} in \texttt{irq0} and the
store of \texttt{x} in \texttt{irq1} is filtered out
when \texttt{irq1}'s priority is not higher than \texttt{irq0}'s
 priority as shown in Fig.~\ref{fig:case2}.

\ignore{

Finally, we state the correctness of our algorithm.
\begin{theorem}
  The priority-aware thread-modular analysis is sound.
\end{theorem}
\begin{proof}
  The process of removing infeasible store-to-load is a form of semantic
  reduction~\cite{Cousot79}  relative to the priority unaware thread-modular
  analysis~\cite{Mine11,Mine12,Mine14}. Specifically, we first divide the
  interferences into an explicit set of interfering stores and their memory
  state, and then we filter any infeasible store-to-load pairs (often referred
  as focusing~\cite{Sagiv99}).
  The analysis determining infeasibility is sound (\textrel{MustNotReadFrom},
  Theorem~\ref{thrm:mnrf}).
  Thus, only truly infeasible store-to-load flows will be removed. So the
  analysis is sound.
\end{proof}

}

Our method can handle programs with loops.  Fig.~\ref{fig:loopmot}
 shows an example, which has two interrupt handlers where 
 \texttt{irq1} has higher priority than \texttt{irq0}. 
Note that \texttt{irq0} loads \texttt{x} and stores the value into \texttt{b}.
Since \texttt{x} is initialized to \texttt{0}, the handler checks 
whether the value of \texttt{b} is \texttt{0}.  
\texttt{irq1} has a loop containing two stores of \texttt{x}. 
 First, it stores the value \texttt{1} and then the value \texttt{0}.  
 Using traditional thread-modular abstract interpretation, we would assume that 
 \texttt{x=1} and \texttt{x=0}  are all possible stores
 to the load of \texttt{x} in \texttt{irq0}.  This would lead to 
 a bogus violation of the assertion in \texttt{irq0}.

However, in our analysis, this bogus violation is avoided by
using the \emph{post-dominate} relation between statements.  Inside the
while-loop of \texttt{irq1}, \texttt{x=0}
post-dominates \texttt{x=1}, meaning that \texttt{x=0} always occurs
after \texttt{x=1}.  Therefore, using our Datalog inference rules presented in the previous
section, we conclude that the store \texttt{x=1} cannot reach the load
of \texttt{x} in \texttt{irq0}.  Thus, it is impossible for
the value \texttt{1} to be stored in \texttt{b} and then cause the assertion violation.

\begin{figure}[!]
\vspace{1ex}
\centering
\hspace{0\linewidth}
  \centering
\framebox[.9\linewidth]{
\begin{minipage}{.9\linewidth}
\centering
\begin{minipage}{.45\linewidth}
{\scriptsize
\texttt{irq0() \{}

\vspace{1ex}

~~~ \texttt{b = x;}

\vspace{1ex}

~~~ \texttt{assert(b == 0)}; 

\vspace{1ex}

\}
}
\end{minipage}
\begin{minipage}{.45\linewidth}
{\scriptsize
\texttt{irq1() \{}

~~~\texttt{while(...) \{}

~~~~~~\texttt{x = 1;}

~~~~~~\texttt{x = 0;}

~~\}

\}
}
\end{minipage}
\end{minipage}
}
\caption{A small example with a loop.}
\label{fig:loopmot}	
\end{figure}

\section{Experiments}
\label{sec:experiment}

We have implemented \Name{}, our new abstract interpretation
framework in a static verification tool for interrupt-driven C
programs.  It builds on a number of open-source tools including
Clang/LLVM~\cite{Adve03} for implementing the C front-end, Apron
library~\cite{Jeannet09} for implementing the abstract domains, and
$\mu Z$~\cite{Hoder11} for solving the Datalog constraints.  We
experimentally compared \Name{} with both iCBMC~\cite{Kroening15}, a
model checker for interrupt-driven programs and the state-of-the-art
thread-modular abstract interpretation method by
Min\'e~\cite{Mine14,Mine11}.  We conducted our experiments on a
computer with an Intel Core i5-3337U CPU, 8 GB of
RAM, and the Ubuntu 14.04 Linux operating system.

Our experiments were designed to answer the following research
questions:
\begin{itemize}
\item 
Can \Name{} prove more properties (e.g., assertions) than
state-of-the-art techniques such as iCBMC~\cite{Kroening15} and
Min\'e~\cite{Mine14,Mine11}?
\item 
Can \Name{} achieve the aforementioned higher accuracy while
maintaining a low computational overhead?
\item
Can \Name{} identify and prune away a large number of infeasible store-load pairs?
\end{itemize}
Toward this end, we evaluated \Name{} on 35 interrupt-driven C
programs, many of which are from real applications such as control
software, firmware, and device drivers.  These benchmark programs,
together with our software tool, have been made available
online~\cite{SaInterrupt-Bench}. The detailed description of each
benchmark group is shown in Table~\ref{tab:benchmark}. In total, there
are 22,541 lines of C code.

\begin{table}
\centering
\footnotesize
\addtolength{\tabcolsep}{-2pt}
\caption{\label{tab:benchmark} Benchmark programs used in our experimental evaluation.}
\scalebox{.9}{
\begin{tabular}{|l|p{.9\linewidth}|} \hline
{\bf Name } & {\bf Description} \\
\hline\hline
test        & Small programs created to conduct the sanity check of \Name{}'s handling of various interrupt semantics.\\\hline
logger      & Programs that model parts of the firmware of a temperature logging device from a major industrial enterprise. There are two major interrupt handlers: one for measurement and the other for communication.  \\\hline
blink       & Programs that control LED lights connected to the MSP430 hardware, to check the timer values and change LED blinking based on the timer values. \\\hline
brake       & Programs generated from the Matlab/Simulink model of a brake-by-wire system from Volvo Technology AB, consisting of a main interrupt handler and four other handlers for computing the braking torque based on the speed of each wheel. \\\hline
usbmouse    & USB mouse driver from the Linux kernel, consisting of the device open, probe, and disconnect tasks with interrupt handlers. \\\hline
usbkbd      & USB keyboard driver from the Linux kernel, consisting of the device open, probe, and disconnect tasks with interrupt handlers. \\\hline 
rgbled      & USB RGB LED driver from the Linux kernel.  We use initialization of \emph{led} and \emph{rgb} functions and the \emph{led} probe function, and check the consistency of the \emph{led} and \emph{rgb} device values using interrupts. \\\hline
rcmain      & Linux device driver for a remote controller core, including operations such as device register, free, check the device information, and update protocol values.  We check the consistency of the device information and protocol values using several interrupt handlers.\\\hline 
others      & Programs collected from Linux kernel drivers for supporting hardware such as ISA boards, TCO timer for i8xx chipsets, and watch dog. \\\hline

\end{tabular}
}
\end{table}

\subsection{Results}

Table~\ref{tbl:resultbox} shows the experimental results.  
Columns~1-4 show the name, the number of lines of code (LoC), the
number of interrupt handlers, and the number of assertions used
 for each benchmark program.
Columns 5-7 show the results of \texttt{iCBMC}~\cite{Kroening15}, including the number
of violations detected, the number of proofs obtained, and the total
execution time.
Columns 8-10 show the results of Min\'e's abstract interpretation method~\cite{Mine14,Mine11}.
Columns 11-13 show the results of \Name{}, our new abstract interpretation tool
for interrupts.

\begin{table*}[ht]
\caption{Results of comparing \Name{} with state-of-the-art techniques on 35 interrupt-driven programs.} 
\label{tbl:resultbox}
\centering
\scalebox{0.965}{
\begin{tabular}{lccccccccccccc}
\toprule
& & & & \multicolumn{3}{c}{iCBMC~\cite{Kroening15}}   
& \multicolumn{3}{c}{Min\'{e}~\cite{Mine14,Mine11}} 
& \multicolumn{3}{c}{intAbs (new)} \\
\cmidrule(r){5-7}
\cmidrule(r){8-10}
\cmidrule(r){11-13}

Name & LOC & Interrupts & Assertions & Violations  & Proofs & Time (s) & \textcolor{black}{Warnings} & Proofs & Time (s)
& \textcolor{black}{Warnings} & Proofs & Time (s) \\
\midrule
test1&46&2 & 2 &0&0&0.23 &1&1&0.18  &0&2&0.07 \\
test2&65&3 & 3 &1&0&0.55 &3&0&0.05  &1&2&0.06 \\
test3&86&4 & 4 &1&0&0.52 &4&0&0.06  &2&2&0.10 \\
test4&56&2 & 2 &1&0&0.52 &2&0&0.04  &1&1&0.05 \\
test5&54&2 & 2 &1&0&1.56 &2&0&0.04  &1&1&0.04 \\
\hline
logger1&161&2 & 1 &0&0&0.45 &1&0&0.22  &0&1&0.27 \\
logger2&183&3 & 3 &0&0&0.50 &2&1&0.29  &0&3&0.39 \\
logger3&195&4 & 4 &0&0&0.46 &1&3&0.31  &0&4&0.43 \\
\hline
blink1&164&3 & 3 &1&0&0.65 &3&0&0.12  &2&1&0.18 \\
blink2&174&4 & 3 &1&0&0.67 &3&0&0.16  &2&1&0.30 \\
blink3&194&5 & 4 &2&0&1.14 &4&0&0.25  &3&1&0.46 \\
\hline
brake1&819&2 & 5 &1&0&0.87 &3&2&0.66  &1&4&0.98 \\
brake2&818&3 & 4 &3&0&2.24 &4&0&1.67  &3&1&1.91 \\
brake3&833&4 & 5 &2&0&2.38 &5&0&2.58  &4&1&3.48 \\
\hline
usbmouse1&426&2 & 8 &2&0&0.79 &7&1&0.11  &2&6&0.13 \\
usbmouse2&442&4 & 16 &2&0&0.69 &16&0&0.31 &5&11&0.69 \\
usbmouse3&449&5 & 20 &11&0&4.00 &20&0&0.52  &11&9&1.28 \\
\hline
usbkbd1&504&2 & 8 &3&0&0.91 &8&0&0.23  &4&4&0.39 \\
usbkbd2&512&3 & 12 &2&0&1.20 &12&0&0.51  &4&8&1.09 \\
usbkbd3&531&5 & 20 &3&0&1.19 &20&0&1.86  &12&8&4.44 \\
\hline
rgbled1&656&2 & 10 &5&0&0.71 &10&0&0.41   &5&5&0.77 \\
rgbled2&679&3 & 15 &5&0&1.11 &15&0&0.99  &5&10&2.39  \\
rgbled3&701&4 & 20 &5&0&1.07 &20&0&2.18  &10&10&5.68 \\
\hline
rcmain1&2060&3 & 9  &0&0&5.36 &9&0&1.58  &0&9&1.80 \\
rcmain2&2088&5 & 15 &6&0&12.39 &15&0&6.93  &6&9&9.46 \\
rcmain3&2102&6 & 18 &9&0&3.95 &18&0&12.20  &9&9&16.35 \\
\hline
i2c\_pca\_isa\_1&321&4 & 6 &0&0&0.41 &6&0&0.14  &0&6&0.29 \\
i2c\_pca\_isa\_2&341&6 & 10 &8&0&2.24 &10&0&0.36  &8&2&1.05 \\
i2c\_pca\_isa\_3&363&8 & 14 &12&0&4.98 &14&0&0.85  &12&2&2.48 \\
i8xx\_tco\_1&757&3 & 2 &0&0&0.30 &2&0&0.28  &0&2&0.35 \\
i8xx\_tco\_2&949&4 & 2 &1&0&0.96 &2&0&0.43  &1&1&0.54 \\
i8xx\_tco\_3&944&6 & 3 &0&0&0.52 &3&0&0.81  &0&3&1.04 \\
wdt\_pci\_1&1239&4 & 2 &0&0& 0.41 &2&0&0.40  &0&2&0.61 \\
wdt\_pci\_2&1290&6 & 3 &0&0& 0.43 &3&0&0.78  &1&2&1.45 \\
wdt\_pci\_3&1339&8 & 4 &0&0& 0.39 &4&0&1.41  &3&1&3.21 \\

\midrule
\textbf{Total} & 22,541 & 136 & 262 & 88 & \textbf{0}  & 56.75 & (254)$^*$ & \textbf{8}  & 39.92 & (118) & \textbf{144} & 64.21 \\

\bottomrule

\end{tabular}
}
 \begin{tablenotes}
      \small
      \item $^*$ indicates the results contain bogus warnings, because the technique was designed for threads, not for interrupts. 
    \end{tablenotes}
\end{table*}

Since \texttt{iCBMC} conducts \emph{bounded} analysis, when it detects
a violation, it is guaranteed to be a real violation; however, when it
does not detect any violation, the property remains undetermined.
Furthermore, since \texttt{iCBMC} by default stops as soon as it
detects a violation, we evaluated it by repeatedly removing the violated
property from the program until it could no longer detect any new
violation.
Also note that since \texttt{iCBMC} requires the user to manually set
up the \emph{interrupt-enabled points} as described in \cite{Kroening15}, during the experiments, we
first ordered the interrupts by priority and then set 
interrupt-enabled points at the beginning of the next interrupt
handler.
For example, given three interrupts \texttt{irq\_L}, \texttt{irq\_M}
and \texttt{irq\_H}, we would set the enabled point of \texttt{irq\_L}
in a main function, the enabled point of \texttt{irq\_M} at the
beginning of \texttt{irq\_L}, and the enabled point of \texttt{irq\_H}
at the beginning of \texttt{irq\_M}.

Overall, \texttt{iCBMC} found 88 violations while obtaining 0 proofs.
\texttt{Min\'e}'s method, which was geared toward proving properties in threads, 
 obtained 8 proofs while reporting 254 warnings, many of which
turned out to be \emph{bogus} warnings.  
%
In contrast, our new method, \Name{}, obtained 144 proofs while
reporting 118 warnings.  This is significantly more accurate than the
prior techniques.  

In terms of the execution time, \Name{} took 64
seconds, which is slightly long than the 39 seconds taken
by \texttt{Min\'e}'s method and the 56 seconds taken by \texttt{iCBMC}.

\subsection{Infeasible Pairs}

Since \Name{} removes infeasible store-load pairs during the
 iterative
analysis of individual interrupt handlers, it tends to spend 
extra time checking the feasibility of these data flow pairs.
Nevertheless, this is the main source of accuracy improvement of \Name{}.
Thus, to understand the trade-off, we have investigated, for each
benchmark program, the total number of store-load pairs and the number
of infeasible store-load pairs identified by our technique.
Table~\ref{tbl:resultPair} summarizes the results, where Column~3
shows the total number of store-load pairs, Column 4 shows the number
of infeasible pairs, and Column 5 shows the percentage.

\begin{table}[ht]
\vspace{1ex}
\caption{Results of total and filtered store-load pairs using \texttt{intAbs}. }
\label{tbl:resultPair}
\centering
\scalebox{0.94}{
\begin{tabular}{lcccc}
\toprule

Name & LOC & \# of Pairs & \# of Filtered Pairs & Filtered Ratio \\
\midrule
test1  &46 &1  &1 & 100\% \\
test2  &65 &4  &2 &  50\% \\
test3  &86 &16 &8 &  50\% \\
test4  &56 &4  &3 &  75\% \\
test5  &54 &4  &3 &  75\% \\
\hline
logger1&161&18 &2 &  11\% \\
logger2&183&32 &6 &  18\% \\
logger3&195&34 &6 &  17\% \\
\hline
blink1&164 &19 &15&  78\% \\
blink2&174 &56 &32&  57\% \\
blink3&194 &120&63&  52\% \\
\hline
brake1&819 &34 &24&  70\% \\
brake2&818 &82 &58&  70\% \\
brake3&833 &164&128&  78\% \\
\hline
usbmouse1&426 &12  &8   &  66\% \\
usbmouse2&442 &168 &136 &  80\% \\
usbmouse3&449 &288 &208 &  72\% \\
\hline
usbkbd1&504   &40  &20  &  50\% \\
usbkbd2&512   &120 &80  &  66\% \\
usbkbd3&531   &400 &280 &  70\% \\
\hline
rgbled1&656   &76  &38  &  50\% \\
rgbled2&679   &228 &152 &  66\% \\
rgbled3&701   &456 &304 &  66\% \\
\hline
rcmain1&2060  &84  &84  & 100\% \\
rcmain2&2088  &560 &476 &  85\% \\
rcmain3&2102  &840 &714 &  85\% \\
\hline
i2c\_pca\_isa\_1&321  &33  &33  & 100\% \\
i2c\_pca\_isa\_2&341  &210 &110 &  52\% \\
i2c\_pca\_isa\_3&363  &434 &240 &  55\% \\
i8xx\_tco\_1&757      &14  &12  &  85\% \\
i8xx\_tco\_2&949      &28  &20  &  74\% \\
i8xx\_tco\_3&944      &39  &33  &  84\% \\
wdt\_pci\_1&1239      &60  &40  &  66\% \\
wdt\_pci\_2&1290      &150 &82  &  54\% \\
wdt\_pci\_3&1339      &288 &139 &  48\% \\
\midrule
\textbf{Total} & 22,541 & 5,116 & 3,560 & 69\% \\

\bottomrule
\end{tabular}
}
\end{table}

Overall, our Datalog-based method for computing
the \textrel{MustNotReadFrom} relation helped remove 69\% of the
load-store pairs, which means the subsequent abstract interpretation
procedure only has to consider the remaining 31\% of the load-store
pairs.  This allows \Name{} to reach a fixed point not only quicker but also with
significantly more accurate results.

\section{Related Work}
\label{sec:relatedwork}

We have reviewed some of the most closely-related work.  In addition,
Min\'e~\cite{Mine17} proposed an abstract interpretation based
technique for proving the absence of data-races, deadlocks, and other
runtime errors in real-time software with dynamic priorities, which is
an extension of his prior work~\cite{Mine12abs,Mine14,Mine11} by adding priorities
while targeting the OSEK/AUTOSAR operating systems. Specifically, it
tracks the effectiveness of mutex, yield and scheduler state based on
execution traces to figure out reachability, while using priorities to
make the analysis more accurate. However, the technique may not be
efficient in terms of memory and speed since it needs to check all
mutex, yield, and scheduler state to determine spurious interference
through trace history. Furthermore, it has not been thoroughly evaluated on
practical benchmarks.

Schwarz and M\"{u}ller-Olm~\cite{Schwarz11} proposed a static analysis
technique for programs synchronized via the priority ceiling protocol.
The goal is to detect synchronization flaws due to concurrency induced
by interrupts, especially for data races and transactional behavior of
procedures.  However, it is not a general-purpose verification
procedure and cannot prove the validity of assertions.
Regehr et al.~\cite{RegehrRW05} proposed to use context-sensitive
abstract interpretation of machine code to guarantee stack-safety for
interrupt-driven programs.  
Kotker and Seshia~\cite{Kotker11} extended a timing analysis procedure
from sequential programs to interrupt-driven programs with a bounded
number of context switches.  As such, it does not analyze all 
behaviors of the interrupts. Furthermore, the user needs to come up
with a proper bound of the context switches and specify the 
arrival time for interrupts.

Wu et al.~\cite{WuWCDW13} leveraged (bound) model checking tools to
detect data-races in interrupt-driven programs.  Kroening et
al.~\cite{Kroening15} also improved the CBMC bounded model checker to support the
verification of interrupt-driven programs. However, they only search
for a bounded number of execution steps, and thus cannot prove the
validity of assertions.
Wu et al.~\cite{WuCMDW16} also proposed a source-to-source
transformation technique similar to Regehr~\cite{Regehr07}: it
sequentializes interrupt-driven programs before feeding them to
a verification tool. However, due to the bounded nature of the
sequentialization process, the method is only suitable for detecting
violations but not for proving the absence of such violations.

This limitation is shared by testing methods.  For example,
Regehr~\cite{Regehr05} proposed a testing framework that schedules the
invocation of interrupts randomly. Higashi and Inoue~\cite{Higashi10}
leveraged a CPU emulator to systematically trigger interrupts to
detect data races in interrupt-driven programs.  However, it may be 
practically infeasible to cover all interleavings of interrupts using
this type of techniques.

Wang et al.~\cite{Wang15,WangWYZL17} proposed a hybrid approach that combines
static program analysis with dynamic simulation to detect data races
in interrupt-driven programs.  Although the approach is useful for
detecting bugs, it cannot be used to obtain proofs, i.e., proofs that assertions
always hold.

There are also formal verification techniques for embedded software based on
model checking~\cite{Ivanci05b,Yang06,Wang06,WangH14,GuoWW17,Schlich09,Vortler15}.
For example, Schlich and Brutschy \cite{Schlich09} proposed the
reduction of interrupt handler points based on partial order reduction
when model checking embedded software.
V{\'o}rtler et al.~\cite{Vortler15} proposed, within the Contiki
system, a method for modeling interrupts at the level of 
hardware-independent C source code and a new modeling approach for
periodically occurring interrupts.  Then, they verify programs with
interrupts using CBMC, which is again a \emph{bounded} model checker.
This means the technique is also geared toward detecting bugs and thus cannot 
prove properties.  Furthermore, since it models the periodical
interrupt invocation only, the approach cannot deal with non-periodic
invocations.

Datalog-based analysis techniques have been widely
used in testing and
verification~\cite{WhaleyL04,LamWLMACU05,LivshitsL05,NaikAW06,BravenboerS09,FarzanK12,GuoKWYG15,GuoKW16},
but none of the prior techniques was designed for analyzing the
interleaving behavior of interrupts. For example, Kusano and
Wang~\cite{KusanoW16,KusanoW17} used Datalog to obtain
flow-sensitivity in threads to improve the accuracy of thread-modular
abstract analysis for concurrent programs.  Sung et al.~\cite{Sung16}
used Datalog-based static analysis of HTML DOM events to speed up a
testing tool for JavaScript based web applications.  However, the
Datalog rules used by these prior techniques cannot be used to reason
about the behavior of interrupts with priorities.

\section{Conclusions}
\label{sec:conclusion}

We have presented an \emph{abstract interpretation} framework
for static verification of interrupt-driven software.  It first
analyzes each individual handler function in isolation and
then propagates the results to other handler functions.  To filter
out the infeasible data flows, we have also developed a
constraint-based analysis of the scheduling semantics of interrupts
with priorities.  It relies on constructing and solving a system of
Datalog constraints to decide whether a set of data flow pairs may
co-exist.  We have implemented our method in a software tool and
evaluated it on a large set of interrupt-driven programs.  Our
experiments show the new method not only is efficient but also
significantly improves the accuracy of the results compared to
existing techniques.  More specifically, it
outperformed both iCBMC, a bounded model checker, and the state-of-the-art
abstract interpretation techniques.

\section{Acknowledgments}

This material is based upon research supported in part by the U.S.\
National Science Foundation under grants CCF-1149454 and CCF-1722710
and the U.S.\ Office of Naval Research under award number
N00014-17-1-2896.

\newpage

\bibliographystyle{plain}
\bibliography{intAbs}

\end{document}